\documentclass[10pt,twocolumn,twoside]{IEEEtran}
\usepackage{amsmath}
\usepackage{color}
\usepackage{cases}
\usepackage{amsfonts,amsmath,amssymb}
\usepackage{mathrsfs}
\usepackage{cite}
\usepackage{float}
\usepackage{graphicx}
\usepackage{algorithm}
\usepackage{algpseudocode}
\usepackage{url}
\usepackage{enumerate}
\usepackage{subfigure}
\usepackage{hyperref}
\usepackage{cite}

\newtheorem{theorem}{Theorem}
\newtheorem{proposition}{Proposition}
\newtheorem{lemma}{Lemma}

\newtheorem{example}{Example}
\newtheorem{remark}{Remark}

\begin{document}

\title{Robust Polynomial Reconstruction via Chinese Remainder Theorem in the Presence of Small Degree Residue Errors}

\author{\mbox{Li Xiao \,\, and  \,\,  Xiang-Gen Xia}
\thanks{The authors are with Department of
Electrical and Computer Engineering,
University of Delaware, Newark, DE 19716, U.S.A. (e-mail:
\{lixiao, xxia\}@ee.udel.edu).}
}
\maketitle

\begin{abstract}
Based on unique decoding of the polynomial residue code with non-pairwise coprime moduli, a polynomial with degree less than that of the least common multiple (lcm) of all the moduli can be accurately reconstructed when the number of residue errors is less than half the minimum distance of the code. However, once the number of residue errors is beyond half the minimum distance of the code, the unique decoding may fail and lead to a large reconstruction error. In this paper, assuming that all the residues are allowed to have errors with small degrees, we consider how to reconstruct the polynomial as accurately as possible in the sense that a reconstructed polynomial is obtained with only the last $\tau$ number of coefficients being possibly erroneous, when the residues are affected by errors with degrees upper bounded by $\tau$. In this regard, we first propose a multi-level robust Chinese remainder theorem (CRT) for polynomials, namely, a trade-off between the dynamic range of the degree of the polynomial to be reconstructed and the residue error bound $\tau$ is formulated. Furthermore, a simple closed-form reconstruction algorithm is also proposed.

\end{abstract}

\section{Introduction}\label{sec1}
Chinese remainder theorem (CRT) is a fundamental result in number theory, which has been widely used in computer architecture, digital signal processing, cryptography, etc. over the past few decades \cite{CRT1,CRT2,CRT3}, as it decomposes a ring of bigger size into several independent rings of smaller sizes. In order to protect systems again errors, error-correcting codes based on the CRT (called residue codes) have been developed for residue error detection and correction in the literature \cite{lly2-2002,lly3-2002,lle-2006,vgoh2008,sundaram,poly1,poly3}. In this paper, motivated by fault-tolerant polynomial-type operations (e.g., cyclic convolution, correlation and FFT computations \cite{peb,yishengsu,mura,haining}) with reduced complexity in digital signal processing systems,
we consider polynomial reconstruction via the CRT for polynomials in the presence of residue errors. It is known that when the number of the residue errors is limited (e.g., less than half the minimum distance of the code), a polynomial with degree less than that of the least common multiple (lcm) of all the moduli can be accurately reconstructed as a unique output in the decoding of the polynomial residue code with non-pairwise coprime moduli \cite{sundaram}. However, if the number of the residue errors is larger (e.g., beyond half the minimum distance of the code), the decoding of the polynomial residue code may fail and lead to a large reconstruction error, and in this case,
how we can reconstruct the polynomial as accurately as possible is of interest in this paper.

In \cite{xiaoxia2}, a robust CRT for polynomials has been studied when all the residues are allowed to have errors but the degrees of the errors are restricted to be small. It basically says that a reconstructed polynomial can be obtained with only the last $\tau$ number of coefficients being possibly erroneous, when the residues are affected by errors with degrees upper bounded by $\tau$, where $\tau$ is called the residue error bound. A sufficient condition for $\tau$ is proposed in \cite{xiaoxia2}. It is not surprising in the robust CRT for polynomials proposed in \cite{xiaoxia2} that the degree of the polynomial to be reconstructed has to be less than the degree of the lcm of all the moduli the same as the CRT for polynomials. In this paper, we propose a multi-level robust CRT for polynomials, which is a generalization of the robust CRT for polynomials proposed in \cite{xiaoxia2}. It reveals a trade-off between the range of the degree of the polynomial to be reconstructed and the residue error bound. In other words, we may increase the residue error bound at the cost of decreasing the dynamic range of the degree of the polynomial to be reconstructed.
Moreover, a simple closed-form reconstruction algorithm is also proposed. Note that there is also some work on the multi-level robust CRT for integers \cite{new,lllxxx}, where the residue vectors of all nonnegative integers less than the dynamic range are presented as points connected by the slanted lines with the slope of $1$ in a high-dimensional space, and a robust reconstruction is obtained by finding the closest point to the erroneous residues on one of the slanted lines. Obviously, this geometric representation method does not work for polynomial cases here. Therefore, the result of this paper is not directly related to that work \cite{new,lllxxx}.

The following notations will be used. Let $\mathbb{F}$ be a field and $\mathbb{F}[x]$ denote the set of all polynomials with coefficients in $\mathbb{F}$ and indeterminate $x$. The highest power of $x$ in a polynomial $f(x)$ is termed the degree of the polynomial, and denoted by $\mbox{deg}\left(f(x)\right)$. All the elements of $\mathbb{F}$ can be expressed as polynomials of degree $0$ and are termed scalars. A polynomial of degree $n$ is called monic if the coefficient of $x^n$ is $1$. We denote the greatest common divisor (gcd) and lcm of a set of polynomials $\{f_i(x)\}_{i=1}^L$ by $\mbox{gcd}\left(f_1(x),\cdots,f_L(x)\right)$ and
$\mbox{lcm}\left(f_1(x),\cdots,f_L(x)\right)$, respectively. For the uniqueness, $\mbox{gcd}(\cdot)$ and $\mbox{lcm}(\cdot)$ are both taken to be monic polynomials. Two polynomials are said to be coprime if their $\mbox{gcd}$ is $1$. The residue of $f(x)$ modulo $g(x)$ is denoted by $\left| f(x)\right|_{g(x)}$. Throughout the paper, all polynomials considered are in $\mathbb{F}[x]$.

\section{Preliminaries}\label{sec2}
Let $m_1(x),\cdots,m_L(x)$ be $L$ non-pairwise coprime moduli. Then, a polynomial $a(x)$ with degree less than that of the lcm of all the moduli can be equivalently represented by its residues $a_i(x)=|a(x)|_{m_i(x)}$ or $a_i(x)\equiv a(x)\mod m_i(x)$, i.e., there exist $k_i(x)$ (called folding polynomials) such that
\begin{equation}\label{folding}
a(x)=k_i(x)m_i(x)+a_i(x)
\end{equation}
with $\mbox{deg}\left(a_i(x)\right)<\mbox{deg}\left(m_i(x)\right)$, for $1\leq i\leq L$.
Conversely, $a(x)$ can be reconstructed from its residues via the CRT for polynomials \cite{CRT1,sundaram}.

If $t$ errors with values $e_{i_1}(x),\cdots,e_{i_t}(x)$ have occurred in the transmission, then the received residues will be given by, for $1\leq i\leq L$,
\begin{equation}
\tilde{a}_i(x) =
  \begin{cases}
    a_i(x)+e_i(x) & \quad \mbox{if } i\in\{i_1,\cdots,i_t\}\\
    a_i(x)  & \quad \mbox{otherwise.}\\
  \end{cases}
\end{equation}
The residue errors $e_i(x)$ satisfy $\mbox{deg}(e_i(x))<\mbox{deg}(m_i(x))$. In \cite{sundaram}, the residue error correction capability in the polynomial residue code with non-pairwise coprime moduli and code distance $d$ has been studied. It is stated in \cite{sundaram} that $a(x)$ with $\mbox{deg}\left(a(x)\right)<\mbox{deg}\left(\mbox{lcm}(m_1(x),\cdots,m_L(x))\right)$ can be accurately reconstructed by the unique decoding when only $t\leq\lfloor(d-1)/2\rfloor$ errors of arbitrary values are in the residues, where $\lfloor\cdot\rfloor$ stands for the floor function. Unfortunately, if the number of residue errors is beyond the error correction capability of the polynomial residue code, the unique decoding may fail and lead to a large reconstruction error. In this case, what we can do is to reconstruct the polynomial as accurately as possible.

In this paper, we are interested in a robust reconstruction problem, that is, a reconstructed polynomial $\hat{a}(x)$ can be obtained such that $\mbox{deg}(\hat{a}(x)-a(x))\leq\tau$ when all residues $a_i(x)$ are allowed to have errors $e_i(x)$ with small degrees upper bounded by $\tau$ (i.e., $\mbox{deg}(e_i(x))\leq\tau$) for $1\leq i\leq L$,
where $\tau$ is called the residue error bound.
Recently, a robust CRT for polynomials has been proposed in \cite{xiaoxia2} as stated in the following result.
\begin{proposition}\cite{xiaoxia2}\label{proposition1}
If the residue error bound $\tau$ satisfies
\begin{equation}\label{upbound}
\tau<\max_{1\leq i\leq L}\min_{1\leq j\leq L; j\neq i}(\mbox{deg}(\mbox{gcd}(m_i(x),m_j(x)))),
\end{equation}
then $a(x)$ with degree less than $\mbox{deg}\left(\mbox{lcm}(m_1(x),\cdots,m_L(x))\right)$ can be robustly reconstructed.
\end{proposition}

A closed-form reconstruction algorithm for Proposition \ref{proposition1} is also proposed in  \cite{xiaoxia2}. For more details, we refer the reader to \cite{xiaoxia2}. According to Proposition \ref{proposition1}, when all the residue errors have small degrees upper bounded by $\tau$ in (\ref{upbound}), we can obtain a reconstructed polynomial $\hat{a}(x)$ for $a(x)$ with degree less than $\mbox{deg}\left(\mbox{lcm}(m_1(x),\cdots,m_L(x))\right)$ such that only the last $\tau$ number of coefficients of $\hat{a}(x)$ may be  different from those of $a(x)$ while the other coefficients are accurately determined. One can see that the dynamic range of the degree of $a(x)$ is $\mbox{deg}\left(\mbox{lcm}(m_1(x),\cdots,m_L(x))\right)$ in the robust CRT for polynomials. In the following, by relaxing the dynamic range of the degree of $a(x)$, we generalize the robust CRT for polynomials in Proposition \ref{proposition1} in a multi-level strategy, called a multi-level robust CRT for polynomials in this paper.

\section{Multi-Level Robust CRT for Polynomials}\label{sec3}
In this section, we first investigate a multi-level robust CRT for polynomials for two-modular systems. A trade-off between the dynamic range of $\mbox{deg}(a(x))$ and the residue error bound $\tau$ is exactly formulated. A simple closed-form reconstruction algorithm is then proposed. By the cascade architecture of CRT as in \cite{lllxxx}, some multi-level result for multi-modular systems can be easily obtained, and thus we skip it in this paper.

Let $m_1(x),m_2(x)$ be two non-coprime polynomial moduli with $\mbox{deg}(m_1(x))\leq\mbox{deg}(m_2(x))$, and $m(x)$ and $M(x)$ be their gcd and lcm, respectively, i.e., $m(x)=\mbox{gcd}\left(m_1(x),m_2(x)\right)$ and $M(x)=\mbox{lcm}\left(m_1(x),m_2(x)\right)$. Since $m(x)$ is the gcd of $m_1(x)$ and $m_2(x)$, we can
write $m_1(x)=m(x)\Gamma_1(x)$ and $m_2(x)=m(x)\Gamma_2(x)$ with $\Gamma_1(x),\Gamma_2(x)$ being coprime.
Let $\sigma_{-1}(x)=\Gamma_2(x),\sigma_0(x)=\Gamma_1(x)$, and for $i\geq1$, we define
\begin{equation}\label{defsigma}
\sigma_i(x)=|\sigma_{i-2}(x)|_{\sigma_{i-1}(x)}.
\end{equation}
It is not hard to see that there must exist an index $K$ with $K\geq0$ such that
\begin{multline}
  \mbox{deg}(\sigma_{-1}(x))>\mbox{deg}(\sigma_{0}(x))>\cdots \\
  >\mbox{deg}(\sigma_{K}(x))>\mbox{deg}(\sigma_{K+1}(x))=0.
\end{multline}

\begin{lemma}\label{longlemma}
Let $a(x)$ be a polynomial with
$\mbox{deg}(a(x))<\mbox{deg}(M(x))-\mbox{deg}(\sigma_{i}(x))$ for some $i$ with $1\leq i\leq K+1$.
If $\mbox{deg}(a_2(x))<\mbox{deg}(m_1(x))\mbox{ and }a_1(x)\neq a_2(x)$,
we have
\begin{equation}\label{lem1}
\mbox{deg}(m(x))+\mbox{deg}(\sigma_{i}(x))\leq\mbox{deg}(a_1(x)-a_2(x))<\mbox{deg}(m_1(x)).
\end{equation}
\end{lemma}
\begin{proof}
It is straightforward that
\begin{align}\label{jinian}
\mbox{deg}(a_1(x)-a_2(x))&\leq\max(\mbox{deg}(a_1(x)),\mbox{deg}(a_2(x)))\nonumber\\
&<\mbox{deg}(m_1(x)).
\end{align}
Thus, we now only need to prove the former inequality in (\ref{lem1}), i.e., $\mbox{deg}(m(x))+\mbox{deg}(\sigma_{i}(x))\leq\mbox{deg}(a_1(x)-a_2(x))$. Based on $\mbox{deg}(a(x))<\mbox{deg}(M(x))-\mbox{deg}(\sigma_{i}(x))$ and (\ref{folding}), it is readily seen that
\begin{equation}\label{degreek2}
\mbox{deg}(k_2(x))<\mbox{deg}(\Gamma_1(x))-\mbox{deg}(\sigma_{i}(x)).
\end{equation}
Furthermore, from (\ref{folding}), we get
\begin{equation}\label{guanxi}
k_2(x)\Gamma_2(x)-k_1(x)\Gamma_1(x)=(a_1(x)-a_2(x))/m(x)
\end{equation}
and then by doing modulo $\Gamma_1(x)$ in both sides of (\ref{guanxi}), we have
\begin{equation}\label{lianxi}
k_2(x)\sigma_{1}(x)\equiv(a_1(x)-a_2(x))/m(x)\mod \Gamma_1(x).
\end{equation}
When $i=1$, i.e., $\mbox{deg}(a(x))<\mbox{deg}(M(x))-\mbox{deg}(\sigma_{1}(x))$, we have, from (\ref{degreek2}),
\begin{equation}\label{eee11}
\mbox{deg}(k_2(x)\sigma_{1}(x))\leq\mbox{deg}(k_2(x))+\mbox{deg}(\sigma_{1}(x))<\mbox{deg}(\Gamma_1(x)).
\end{equation}
From (\ref{jinian}), (\ref{lianxi}), (\ref{eee11}) and $a_1(x)\neq a_2(x)$, we have $a_1(x)-a_2(x)=k_2(x)\sigma_{1}(x)m(x)\mbox{ with }k_2(x)\neq0$,
and hence $\mbox{deg}(m(x))+\mbox{deg}(\sigma_{1}(x))\leq\mbox{deg}(a_1(x)-a_2(x))$.
When $i\geq2$, i.e., $\mbox{deg}(a(x))<\mbox{deg}(M(x))-\mbox{deg}(\sigma_{i}(x))$, we have
\begin{equation}\label{eee22}
\mbox{deg}(k_2(x)\sigma_{1}(x))<\mbox{deg}(\Gamma_1(x))+\mbox{deg}(\sigma_{1}(x))-\mbox{deg}(\sigma_{i}(x)).
\end{equation}
From (\ref{jinian}), (\ref{lianxi}) and (\ref{eee22}), there exists a polynomial $c_1(x)$ with $\mbox{deg}(c_1(x))<\mbox{deg}(\sigma_{1}(x))-\mbox{deg}(\sigma_{i}(x))$ such that
\begin{equation}\label{zhankai}
k_2(x)m(x)\sigma_1(x)=c_1(x)m(x)\Gamma_1(x)+a_1(x)-a_2(x).
\end{equation}
If $c_1(x)=0$, we have $a_1(x)-a_2(x)=k_2(x)m(x)\sigma_1(x)$, and from $a_1\neq a_2(x)$, we get $\mbox{deg}(a_1(x)-a_2(x))\geq\mbox{deg}(m(x))+\mbox{deg}(\sigma_1(x))>\mbox{deg}(m(x))+\mbox{deg}(\sigma_i(x))$. If $c_1(x)\neq0$, since
$\Gamma_1(x)=d_1(x)\sigma_1(x)+\sigma_2(x)$ holds for some $d_1(x)$ with $\mbox{deg}(d_1(x))\geq1$ from (\ref{defsigma}), (\ref{zhankai}) becomes
\begin{align}\label{zk1}
k_2(x)m(x)\sigma_1(x)&=c_1(x)d_1(x)m(x)\sigma_1(x)\nonumber\\
&+c_1(x)m(x)\sigma_2(x)+a_1(x)-a_2(x).
\end{align}
Let $r_1(x)=|a_1(x)-a_2(x)|_{m(x)\sigma_1(x)}$, then by doing modulo $m(x)\sigma_1(x)$ in both sides of (\ref{zk1}), we have
\begin{equation}\label{enene}
-c_1(x)m(x)\sigma_2(x)\equiv r_1(x)\mod m(x)\sigma_1(x).
\end{equation}
Due to $\mbox{deg}(c_1(x))<\mbox{deg}(\sigma_{1}(x))-\mbox{deg}(\sigma_{i}(x))$ and (\ref{enene}), there exists a polynomial $c_2(x)$ with $\mbox{deg}(c_2(x))<\mbox{deg}(\sigma_{2}(x))-\mbox{deg}(\sigma_{i}(x))$ such that
\begin{equation}
-c_1(x)m(x)\sigma_2(x)=c_2(x)m(x)\sigma_1(x)+r_1(x).
\end{equation}
If $c_2(x)=0$, we have $r_1(x)=|a_1(x)-a_2(x)|_{m(x)\sigma_1(x)}=-c_1(x)m(x)\sigma_2(x)$, and we can get $\mbox{deg}(a_1(x)-a_2(x))\geq\mbox{deg}(m(x))+\mbox{deg}(\sigma_2(x))>\mbox{deg}(m(x))+\mbox{deg}(\sigma_i(x))$. If $c_2(x)\neq0$, similar to (\ref{zk1}) and (\ref{enene}),
since $\sigma_1(x)=d_2(x)\sigma_2(x)+\sigma_3(x)$ holds for some $d_2(x)$ with $\mbox{deg}(d_1(x))\geq1$ from (\ref{defsigma}), let $r_2(x)=|r_1(x)|_{m(x)\sigma_2(x)}$, and we have
\begin{equation}
-c_2(x)m(x)\sigma_3(x)\equiv r_2(x)\mod m(x)\sigma_2(x).
\end{equation}
Continuing this procedure, let $r_{i-1}(x)=|r_{i-2}(x)|_{m(x)\sigma_{i-1}(x)}$, and we can have, for some $c_{i-1}(x)\neq0$ with $\mbox{deg}(c_{i-1}(x))<\mbox{deg}(\sigma_{i-1}(x))-\mbox{deg}(\sigma_i(x))$,
\begin{equation}\label{xiaozhan}
-c_{i-1}(x)m(x)\sigma_i(x)\equiv r_{i-1}(x)\mod m(x)\sigma_{i-1}(x).
\end{equation}
If $\mbox{deg}(a_1(x)-a_2(x))<\mbox{deg}(m(x))+\mbox{deg}(\sigma_i(x))$, we have $r_1(x)=r_2(x)=\cdots=r_{i-1}(x)=a_1(x)-a_2(x)$. Moreover, from $\mbox{deg}(-c_{i-1}(x)m(x)\sigma_i(x))<\mbox{deg}(m(x)\sigma_{i-1}(x))$, $a_1(x)\neq a_2(x)$ and (\ref{xiaozhan}), we get $a_1(x)-a_2(x)=r_{i-1}(x)=-c_{i-1}(x)m(x)\sigma_i(x)$ with $c_{i-1}(x)\neq0$, and thereby $\mbox{deg}(a_1(x)-a_2(x))\geq\mbox{deg}(m(x))+\mbox{deg}(\sigma_i(x))$, which contradicts our assumption that $\mbox{deg}(a_1(x)-a_2(x))<\mbox{deg}(m(x))+\mbox{deg}(\sigma_i(x))$. Hence, we demonstrate that $\mbox{deg}(a_1(x)-a_2(x))\geq\mbox{deg}(m(x))+\mbox{deg}(\sigma_i(x))$.
\end{proof}

When the residues $a_1(x),a_2(x)$ have small degree errors with the residue error bound $\tau$, i.e.,
$\mbox{deg}(e_i(x))\leq\tau$ for $i=1,2$,
let $\tilde{a}_1(x),\tilde{a}_2(x)$ be the two erroneous residues, and we define
\begin{equation}\label{noseque}
q_{21}(x)=\tilde{a}_1(x)-\tilde{a}_2(x).
\end{equation}
Based on Lemma \ref{longlemma}, we then have the following result.
\begin{lemma}\label{xiaolem}
Let $a(x)$ be a polynomial with $\mbox{deg}(a(x))<\mbox{deg}(M(x))-\mbox{deg}(\sigma_{i}(x))$ for some $i$ with $1\leq i\leq K+1$, and the residue error bound $\tau$ satisfy
\begin{equation}\label{condition1}
\tau<\mbox{deg}(m(x))+\mbox{deg}(\sigma_{i}(x)).
\end{equation}
We can obtain the following three cases:
\begin{itemize}
  \item[$1)$] if $\mbox{deg}(m_1(x))>\mbox{deg}(q_{21}(x))\geq\mbox{deg}(m(x))+\mbox{deg}(\sigma_{i}(x))$, we have  $\mbox{deg}(a_2(x))<\mbox{deg}(m_1(x))$ and $a_1(x)\neq a_2(x)$;
  \item[$2)$] if $\mbox{deg}(q_{21}(x))\geq\mbox{deg}(m_1(x))$, we have $\mbox{deg}(a_2(x))\geq\mbox{deg}(m_1(x))$;
  \item[$3)$] if $\mbox{deg}(q_{21}(x))<\mbox{deg}(m(x))+\mbox{deg}(\sigma_{i}(x))$, we have $a_1(x)=a_2(x)$.
\end{itemize}
\end{lemma}
\begin{proof}
Based on the method of proof by contradiction, we give an indirect proof as follows.
When $\mbox{deg}(a_2(x))<\mbox{deg}(m_1(x))$ and $a_1(x)\neq a_2(x)$, from Lemma \ref{longlemma} we have (\ref{lem1}).
So, from (\ref{condition1}) and (\ref{noseque}) we get
$\mbox{deg}(m_1(x))>\mbox{deg}(q_{21}(x))\geq\mbox{deg}(m(x))+\mbox{deg}(\sigma_{i}(x))$.
When $\mbox{deg}(m_1(x))\leq\mbox{deg}(a_2(x))$, we have $\mbox{deg}(a_1(x)-a_2(x))\geq\mbox{deg}(m_1(x))$.
So, from (\ref{condition1}) and (\ref{noseque}) we get $\mbox{deg}(q_{21}(x))\geq\mbox{deg}(m_1(x))$.
When $a_1(x)=a_2(x)$, we have $q_{21}(x)=\tilde{a}_1(x)-\tilde{a}_2(x)=e_1(x)-e_2(x)$, and from (\ref{condition1}) we get
$\mbox{deg}(q_{21}(x))=\mbox{deg}\left(e_1(x)-e_2(x)\right)<\mbox{deg}(m(x))+\mbox{deg}(\sigma_{i}(x))$.
This completes the proof.
\end{proof}

In the following, we propose a simple closed-form reconstruction algorithm for a polynomial $a(x)$ with $\mbox{deg}(a(x))<\mbox{deg}(M(x))-\mbox{deg}(\sigma_{i}(x))$ for some $i$ with $1\leq i\leq K+1$ such that it can be robustly reconstructed from the erroneous residues $\tilde{a}_1(x),\tilde{a}_2(x)$ with the error bound in (\ref{condition1}).

\begin{algorithm}[H]
  \caption{\!:}
  \label{Ployalg:Framwork1}
  \begin{algorithmic}[1]
  \State Calculate $q_{21}(x)=\tilde{a}_1(x)-\tilde{a}_2(x)$ as in (\ref{noseque}).
   \label{code:11}
  \State (\romannumeral1) If $\mbox{deg}(m_1(x))>\mbox{deg}(q_{21}(x))\geq\mbox{deg}(m(x))+\mbox{deg}(\sigma_{i}(x))$, calculate $f_1(x)=|q_{21}(x)|_{m(x)\sigma_1(x)}$, $f_2(x)=|f_1(x)|_{m(x)\sigma_2(x)}$, $\cdots$, $f_i(x)=|f_{i-1}(x)|_{m(x)\sigma_{i}(x)}$. Then, let
  \begin{equation}
  \hat{k}_2(x)\equiv\frac{q_{21}(x)-f_i(x)}{m(x)}\bar{\Gamma}_{21}(x)\mod \Gamma_1(x),
  \end{equation}
  where $\bar{\Gamma}_{21}(x)$ is the modular multiplicative inverse of $\Gamma_2(x)$ modulo $\Gamma_1(x)$, i.e., $1\equiv\bar{\Gamma}_{21}(x)\Gamma_2(x)\mod \Gamma_1(x)$.

  \noindent(\romannumeral2) If $\mbox{deg}(q_{21}(x))\geq\mbox{deg}(m_1(x))$, calculate $g_0(x)=|q_{21}(x)|_{m_1(x)}$, $g_1(x)=|g_0(x)|_{m(x)\sigma_1(x)}$, $g_2(x)=|g_1(x)|_{m(x)\sigma_2(x)}$, $\cdots$, $g_i(x)=|g_{i-1}(x)|_{m(x)\sigma_i(x)}$. Then, let
  \begin{equation}
 \hat{k}_2(x)\equiv\frac{q_{21}(x)-g_i(x)}{m(x)}\bar{\Gamma}_{21}(x)\mod \Gamma_1(x).
  \end{equation}

  \noindent(\romannumeral3) Otherwise, let $\hat{k}_2(x)=0$.
   \label{code:21}
   \State Calculate
   \begin{equation}\label{rce}
   \hat{a}(x)=\hat{k}_2(x)m_2(x)+\tilde{a}_2(x).
   \end{equation}
  \end{algorithmic}
\end{algorithm}

\begin{theorem}\label{maintheo}
Let $a(x)$ be a polynomial with $\mbox{deg}(a(x))<\mbox{deg}(M(x))-\mbox{deg}(\sigma_{i}(x))$ for some $i$ with $1\leq i\leq K+1$. If the residue error bound $\tau$ satisfies
\begin{equation}\label{newcond}
\tau<\mbox{deg}(m(x))+\mbox{deg}(\sigma_{i}(x)),
\end{equation}
then by \textbf{Algorithm \ref{Ployalg:Framwork1}} we can accurately determine the folding polynomial $k_2(x)$, i.e., $\hat{k}_2(x)=k_2(x)$, and thus $a(x)$ can be robustly reconstructed from the erroneous residues.
\end{theorem}
\begin{proof}
(\romannumeral1) If $\mbox{deg}(m_1(x))>\mbox{deg}(q_{21}(x))\geq\mbox{deg}(m(x))+\mbox{deg}(\sigma_{i}(x))$, from Lemma \ref{xiaolem} we get $\mbox{deg}(a_2(x))<\mbox{deg}(m_1(x))$ and $a_1(x)\neq a_2(x)$. Based on (\ref{noseque}) and (\ref{newcond}), one can see that $f_1(x)=|q_{21}(x)|_{m(x)\sigma_1(x)}=|a_1(x)-a_2(x)|_{m(x)\sigma_1(x)}+e_1(x)-e_2(x)=r_1(x)+e_1(x)-e_2(x)$, $f_2(x)=r_2(x)+e_1(x)-e_2(x)$, $\cdots$, $f_i(x)=r_i(x)+e_1(x)-e_2(x)$. Following the procedure (\ref{zhankai})-(\ref{xiaozhan}) in the proof of Lemma \ref{longlemma}, there must exist an index $j$ with $1\leq j\leq i$ such that
$r_{j}(x)=r_{j+1}(x)=\cdots=r_{i}(x)=0$.
Therefore, we can have $f_i(x)=e_1(x)-e_2(x)$.
Then, it follows from (\ref{newcond}) that $a_1(x)-a_2(x)=q_{21}(x)-f_i(x)$.
By the B\'{e}zout's lemma for polynomials \cite{CRT3} in (\ref{guanxi}), the folding polynomial $k_2(x)$ can be determined by
\begin{equation}
\begin{split}
k_2(x) & \equiv \frac{a_1(x)-a_2(x)}{m(x)}\bar{\Gamma}_{21}(x)\mod \Gamma_1(x) \\
 & \equiv \frac{q_{21}(x)-f_i(x)}{m(x)}\bar{\Gamma}_{21}(x)\mod \Gamma_1(x),
\end{split}
\end{equation}
where $\bar{\Gamma}_{21}(x)$ is the modular multiplicative inverse of $\Gamma_2(x)$ modulo $\Gamma_1(x)$. Hence, $\hat{k}_2(x)=k_2(x)$.

(\romannumeral2) If $\mbox{deg}(q_{21}(x))\geq\mbox{deg}(m_1(x))$, from Lemma \ref{xiaolem} we get $\mbox{deg}(a_2(x))\geq\mbox{deg}(m_1(x))$. When $i=1$, i.e., $\mbox{deg}(a(x))<\mbox{deg}(M(x))-\mbox{deg}(\sigma_{1}(x))$, we obtain from (\ref{lianxi}) and (\ref{eee11}) that
$a_1(x)-a_2(x)=c_0(x)m(x)\Gamma_1(x)+k_2(x)m(x)\sigma_1(x)\mbox{ with }c_0(x)\neq0$.
Let $h_0(x)=|a_1(x)-a_2(x)|_{m(x)\Gamma_1(x)}$ and $h_1(x)=|h_0(x)|_{m(x)\sigma_1(x)}$. Then, from $\mbox{deg}(k_2(x)m(x)\sigma_1(x))<\mbox{deg}(m(x)\Gamma_1(x))$, we have $h_0(x)=k_2(x)m(x)\sigma_1(x)$ and $h_1(x)=0$. When $i\geq2$, i.e., $\mbox{deg}(a(x))<\mbox{deg}(M(x))-\mbox{deg}(\sigma_{i}(x))$, we have
$h_0(x)\equiv k_2(x)m(x)\sigma_1(x)\mod m(x)\Gamma_1(x)$.
So, there exists a polynomial $c_1(x)$ with $\mbox{deg}(c_1(x))<\mbox{deg}(\sigma_{1}(x))-\mbox{deg}(\sigma_{i}(x))$ such that
$k_2(x)m(x)\sigma_1(x)=c_1(x)m(x)\Gamma_1(x)+h_0(x)$.
Then, similar to the analysis of (\ref{zhankai})-(\ref{xiaozhan}) in the proof of Lemma \ref{longlemma}, there must exist an index $j$ with $1\leq j\leq i$ such that $h_{j}(x)=h_{j+1}(x)=\cdots=h_{i}(x)=0$,
where $h_0(x)=|a_1(x)-a_2(x)|_{m(x)\Gamma_1(x)}$, $h_1(x)=|h_0(x)|_{m(x)\sigma_1(x)}$, $\cdots$, $h_i(x)=|h_{i-1}(x)|_{m(x)\sigma_i(x)}$. Moreover, based on (\ref{noseque}) and (\ref{newcond}), one can see that $g_0(x)=|q_{21}(x)|_{m_1(x)}=|a_1(x)-a_2(x)|_{m_1(x)}+e_1(x)-e_2(x)=h_0(x)+e_1(x)-e_2(x)$, $g_1(x)=h_1(x)+e_1(x)-e_2(x)$, $g_2(x)=h_2(x)+e_1(x)-e_2(x)$, $\cdots$, $g_i(x)=h_i(x)+e_1(x)-e_2(x)$.
Therefore, we can have $g_i(x)=e_1(x)-e_2(x)$.
Then, it follows from (\ref{newcond}) that $a_1(x)-a_2(x)=q_{21}(x)-g_i(x)$.
By the B\'{e}zout's lemma for polynomials in (\ref{guanxi}), the folding polynomial $k_2(x)$ can be determined by
\begin{equation}
\begin{split}
k_2(x) & \equiv \frac{a_1(x)-a_2(x)}{m(x)}\bar{\Gamma}_{21}(x)\mod \Gamma_1(x) \\
 & \equiv \frac{q_{21}(x)-g_i(x)}{m(x)}\bar{\Gamma}_{21}(x)\mod \Gamma_1(x).
\end{split}
\end{equation}
Hence, $\hat{k}_2(x)=k_2(x)$.

(\romannumeral3) Otherwise, we know $a_1(x)=a_2(x)$ from Lemma \ref{xiaolem}. In this case, it is immediate to have $k_2(x)=0$, and thereby  $\hat{k}_2(x)=k_2(x)$=0.

Therefore, by \textbf{Algorithm \ref{Ployalg:Framwork1}} the folding polynomial $k_2(x)$ can be accurately determined, and $a(x)$ can be robustly reconstructed as in (\ref{rce}) from the erroneous residues.
\end{proof}

\begin{remark}
Due to $\mbox{deg}(\sigma_{K+1}(x))=0$, Theorem \ref{maintheo} coincides with Proposition \ref{proposition1} in the case of two moduli. When the dynamic range attains the maximum, i.e., the degree of the lcm of the moduli, the residue error bound $\tau$ decreases to the degree of the gcd of the moduli.
\end{remark}

To well illustrate Theorem \ref{maintheo} and \textbf{Algorithm \ref{Ployalg:Framwork1}}, we next present an example.
\begin{example}
Let $\mathbb{F}_2$ denote the finite field with two elements $0$ and $1$, and $\mathbb{F}_2[x]$ denote the set of all polynomials with coefficients in $\mathbb{F}_2$. Let $m_1(x)=(x^2+1)(x^6+x^3+1)$ and $m_2(x)=(x^2+1)(x^9+x^7+x+1)$ be two polynomial moduli in $\mathbb{F}_2[x]$. The gcd of the two moduli is $m(x)=x^2+1$, $\Gamma_1(x)=x^6+x^3+1$, $\Gamma_2(x)=x^9+x^7+x+1$, and the degree of the lcm of the two moduli is $\mbox{deg}(M(x))=17$. By (\ref{defsigma}), we have $\sigma_1(x)=x^4$, $\sigma_2(x)=x^3+1$, $\sigma_3(x)=x$ and $\sigma_4(x)=1$. According to Theorem \ref{maintheo}, we have the following result in Table \ref{table1}, where the last row, i.e., Level \uppercase\expandafter{\romannumeral1}, is the known result in Proposition \ref{proposition1}.
More specifically, given a polynomial $a(x)=x^{15}+x^{11}+x^7+x^6+x+1$, its two residues and the folding polynomial $k_2(x)$ can be calculated as $a_1(x)=x^7+x^2+x+1$, $a_2(x)=x^5+x^4+x+1$ and $k_2(x)=x^4$, respectively. We next use \textbf{Algorithm \ref{Ployalg:Framwork1}} for $i=3$ (or Level \uppercase\expandafter{\romannumeral2} in Table \ref{table1}) in Theorem \ref{maintheo} to robustly reconstruct $a(x)$ when the erroneous residues are given by
\begin{equation}
\tilde{a}_1(x)=x^7+x^2+x+1+\underbrace{x^2+x+1}_{e_1(x)}=x^7
\end{equation}
and
\begin{equation}
\tilde{a}_2(x)=x^5+x^4+x+1+\underbrace{x}_{e_2(x)}=x^5+x^4+1,
\end{equation}
where $\max(\mbox{deg}(e_1(x)),\mbox{deg}(e_2(x)))<3$ makes the condition (\ref{newcond}) hold.
\begin{itemize}
  \item Calculate $q_{21}(x)=x^7+x^5+x^4+1$.
  \item Due to $8=\mbox{deg}(m_{1}(x))>\mbox{deg}(q_{21}(x))=7\geq\mbox{deg}(m(x))+\mbox{deg}(\sigma_3(x))=3$, we calculate $f_1(x)=|q_{21}(x)|_{m(x)\sigma_1(x)}=x^4+1$, $f_2(x)=|f_1(x)|_{m(x)\sigma_2(x)}=x^4+1$, $f_3(x)=|f_2(x)|_{m(x)\sigma_3(x)}=x^2+1$. Then, we have
      \begin{equation}
      \begin{split}
      \hat{k}_2(x)&\equiv\frac{q_{21}(x)-f_3(x)}{m(x)}\bar{\Gamma}_{21}(x)\mod \Gamma_1(x)\\
      &\equiv (x^5+x^2)x^5\mod x^6+x^3+1\\
      &\equiv x^4\mod x^6+x^3+1.
      \end{split}
      \end{equation}
      So, we obtain $\hat{k}_2(x)=x^4$.
  \item We get $\hat{a}(x)=\hat{k}_2(x)m_2(x)+\tilde{a}_2(x)=x^{15}+x^{11}+x^7+x^6+1$.
\end{itemize}
Clearly, the above obtained $\hat{a}(x)$ is a robust estimation of $a(x)$, i.e., $\mbox{deg}(\hat{a}(x)-a(x))=1<3$.

\newcounter{mytempeqncnt}
\begin{table}[!h]
\caption{The trade-off between the dynamic range and the residue error bound.}\label{table1}
\centering
\begin{tabular*}{0.5\textwidth}{@{\extracolsep{\fill}}cccc}  
\hline
level& degree of $\sigma_i(x)$ &residue error bound&dynamic range\\ \hline  
\uppercase\expandafter{\romannumeral4}&$\mbox{deg}(\sigma_1(x))=4$ &$\tau<2+4=6$ &$\mbox{deg}(a(x))<17-4=13$ \\         
\uppercase\expandafter{\romannumeral3}&$\mbox{deg}(\sigma_2(x))=3$ &$\tau<2+3=5$ &$\mbox{deg}(a(x))<17-3=14$\\        
\uppercase\expandafter{\romannumeral2}&$\mbox{deg}(\sigma_3(x))=1$ &$\tau<2+1=3$ &$\mbox{deg}(a(x))<17-1=16$ \\
\uppercase\expandafter{\romannumeral1}&$\mbox{deg}(\sigma_4(x))=0$ &$\tau<2+0=2$ &$\mbox{deg}(a(x))<17-0=17$ \\
\hline
\end{tabular*}
\end{table}
\end{example}

\section{Conclusion}
In this paper, we studied a multi-level robust CRT for polynomials, which shows that a polynomial with different levels of dynamic ranges can be robustly reconstructed from the erroneous residues with correspondingly different levels of error bounds. Furthermore, a simple closed-form reconstruction algorithm was proposed. 

\end{document}